\documentclass[times, 10pt,twocolumn]{article} 
\usepackage{latex8}
\usepackage{times}
\usepackage{enumerate,amsfonts,graphicx,url,stfloats,amsthm}
\usepackage[english]{babel}

\newtheorem{theorem}{Theorem}
\newtheorem{lemma}{Lemma}
\newtheorem{definition}{Definition}

\newcommand{\struktura}{{\mathbb D}}
\newcommand{\odcinki}{{\mathrm I}(\struktura)}

\newcommand{\todo}[1]{\framebox[1.1\width]{\bf Unfinished}}

\author{Draft (\today) \\\bigskip \\
Jerzy Marcinkowski,  Jakub Michaliszyn\\
Institute of Computer Science,\\ University Of Wroclaw,\\
ul. Joliot-Curie 15, 50-383 Wroclaw, Poland\\
\url{{jma,jmi}@cs.uni.wroc.pl}
}

\title{The Last Paper on the Halpern--Shoham  Interval Temporal Logic}

\begin{document}
\maketitle

\begin{abstract}
The Halpern--Shoham logic is a  modal logic of time intervals.
Some effort has been put in last ten years to classify fragments of this beautiful 
 logic with respect to decidability of its satisfiability problem.
 We contribute to this effort by showing -- what we believe is quite an unexpected result -- 
that the  logic of subintervals, the fragment of the Halpern--Shoham  where only the operator ``during'', or $D$, is allowed, 
is undecidable over discrete structures. 
This is surprising as this logic is decidable over dense orders \cite{APS09} and its reflexive variant
 is known to be  decidable over discrete structures \cite{APS10}. 
Our result subsumes a lot of previous results for the discrete case, 
like the undecidability for $ABE$ \cite{HS91}, $BE$ \cite{Lo00}, $BD$ \cite{MM09}, $ADB$, $A\bar{A}D$, and so on \cite{BDG08,BGM09}.
\end{abstract}

\section{Introduction}

In classical temporal logic structures are defined by assigning
properties (propositional variables) to points of time (which is an
ordering, discrete or dense). However, not all  phenomena  can
be well described by such logics. Sometimes we need to talk about
actions (processes) that take some time and we would like to be able
to say that one such action takes place, for example, during or after another.

The Halpern--Shoham logic \cite {HS91}, which is the subject of this
paper, is one of the modal logics of time intervals.
Judging by the number of papers published, and by the amount of work
devoted to the research on it, this logic is probably the
most influential of time interval logics. But historically it was not
the first one. Actually, the earliest papers about intervals in
context of modal logic were written  by philosophers, e.g.,
\cite{Ha71}. In computer science,  the earliest attempts to formalize time intervals were process logic \cite{Par78,Pra79} and  interval temporal logic  \cite{Mo83}. Relations
between intervals in linear orders from an algebraic point of view were
first studied systematically by Allen \cite{Al83}.

The Halpern--Shoham logic is a modal temporal logic, where the elements
of a model are no longer --- like in classical temporal logics ---
points in time, but rather pairs of points in time. Any such pair --- call it
$[p, q]$, where $q$ is not earlier than $p$  --- can be
viewed as a (closed) time interval, that is, the set of  all time
points between $p$ and $q$. HS logic does not assume anything about
order --- it can be discrete or continuous, linear or branching,
complete or not.

 Halpern and Shoham introduce six modal operators, acting on
intervals. Their operators are ``begins'' $B$, ``during'' $D$, ``ends'' $E$, ``meets'' $A$, ``later'' $L$, ``overlaps'' $O$ and the six inverses of those operators: $\bar{B}, \bar{D}, \bar{E}, \bar{A},
\bar{L}$, $\bar{O}$.
It is easy to see that the set of operators is redundant. The ,,more expressive'' of them, which are $A,B$ and $E$   can define $D$ 
($B$ and $E$ suffice for that -- a prefix of my suffix is my infix)
and $L$ (here $A$ is enough --``later'' means ``meets an interval that meets''). The operator $O$ can be expressed using $E$ and $\bar{B}$.

In their paper, Halpern and Shoham show that (satisfiability of
formulae of) their logic is undecidable. Their proof requires logic with 
five operators ($B, E$ and $A$ are explicitly used in the formulae and, as
we mentioned above, once  $B, E$ and $A$ are allowed, $D$ and $L$ come
for free) so they state a question  about decidable fragments of their logic.

Considerable  effort has been put since this time  to settle this question.  First,
 it was shown \cite{Lo00} that the  $BE$ fragment is undecidable. Recently, negative results were also given  
 for the classes $B\bar{E}$, $\bar{BE}$, $\bar{B}E$,
$A\bar{A}D, \bar{A}D^*\bar{B}, \bar{A}D^*B$
\cite{BDG08,BGM09}, and $BD$ \cite{MM09}. Another elegant negative result was that $O\bar{O}$ is undecidable over discrete orders \cite{BDG09,BDG10}.

On the positive side, it was shown that some small fragments, like
$B\bar{B}$ or $E\bar{E}$, are decidable and easy to
translate into standard, point-based modal logic \cite{Gm04}. The
fragment using only $A$ and $\bar{A}$ is a bit harder and its
decidability was only recently  shown \cite{BGM09,BMS08}.
Obviously, this last  result implies decidability of $L\bar{L}$
as $L$ is expressible by $A$. Another fragment known to be decidable is $AB\bar{B}$ \cite{MPS09}.

The last interesting fragment of the Halpern and Shoham logic of unknown status was the, apparently very simple,
 fragment with the single operator $D$ (,,during''), which we call
here {\em the logic of sub-intervals}. Since $D$ 
does not seem to have much expressive power (an example of a formula would be ,,each morning I spend a while thinking of you'' or 
,,each nice period of my life contains an unpleasant fragment'') logic of sub-intervals  was widely believed to be decidable. A number of decidability
results concerning variants of this logic has been published. For example, it was shown in (\cite{BGM08,APS09}) than satisfiability of formulae 
of logic of subintervals is decidable over dense structures. In  \cite{APS10}.  decidability is proved for (slightly less expressive) ,,reflexive $D$''.
The results in \cite{SZ10}  imply that $D$ (as well as some richer fragments of the HS logic) is decidable if we allow models, in which not
all the intervals defined by the ordering are elements of the Kripke structure.

In this paper we show that satisfiability of formulae from 
the $D$ fragment is undecidable over the class of finite orderings as well as over the class of all
 discrete orderings. Our result subsumes the negative results  for the discrete case for  $ABE$ \cite{HS91}, $BE$ \cite{Lo00}, $BD$ \cite{MM09} and
$ADB$, $A\bar{A}D$  \cite{BDG08,BGM09}.

\subsection{Main theorems}

Our contribution consists of the proofs of the following two theorems:

\begin{theorem}\label{main}
The satisfiability problem for the formulae of the  logic  of subintervals, over models which are suborders of the order 
$\langle {\mathbf Z}, \leq \rangle$, is undecidable.
\end{theorem}

Since truth value of a formula is defined with respect to a model and an initial interval in this model (see Preliminaries), and since the only allowed operator is $D$, which means that the truth value of a formula in a given interval depends only on the labeling of this interval and its subintervals
Theorem \ref{main} can be restated as: {\em The satisfiability problem for the formulae of the  logic  of subintervals, 
over finite models  is undecidable}, and it is this version that will be proved in Section \ref{natural}  .

\begin{theorem}\label{main-discrete}
The satisfiability problem for the formulae of the logic of sub-intervals, over all discrete models,  is undecidable.
\end{theorem}

\section{Preliminaries} \label{semantykad}

\noindent{\bf Orderings.} As in \cite {HS91}, we say that a total order $\langle \struktura, \leq \rangle$ is  \emph{discrete} if each element is either minimal (maximal) or has a unique predecessor (successor); in other words for all $a, b \in \struktura$ if $a<b$, then there exist points $a', b'$ such that $a < a'$, $b' <  b$ and there exists no $c$ with $a < c < a'$ or $b' < c < b $.

\vspace{10pt}
\noindent
{\bf Semantic of the $D$ fragment of logic HS (logic of sub-intervals).}
Let $\langle \struktura, \leq\rangle$ be a discrete ordered set
\footnote{To keep the notation light, we will identify the order $\langle \struktura, \leq \rangle$ with its set $\struktura$}.

 An \emph{interval} over $\struktura$ is a pair $[a, b]$ with $a, b \in \struktura$ and $a \leq b$. 
A \emph{labeling} is a function $\gamma : \odcinki \to {\cal P}({\cal V}ar)$, 
where $\odcinki$ is a set of all intervals over $\struktura$ and ${\cal V}ar$ is a finite set of variables. 
A structure of the form $\mathrm{M}=\langle \odcinki, \gamma \rangle$ is called a \emph{model}.

 We say that an interval $[a, b]$ is a \emph{leaf} iff it has no sub-intervals (i.e. $a=b$). 
 
The truth values of formulae are determined by the following (natural) semantic rules:

\begin{enumerate}
\item For all $v \in {\cal V}ar$ we have $\mathrm{M}, [a, b] \models v$ iff $v \in \gamma([a, b])$.
\item $\mathrm{M}, [a, b] \models \neg \varphi$ iff $\mathrm{M},[a, b] \not \models \varphi$.
\item $\mathrm{M}, [a, b] \models \varphi_1 \wedge \varphi_2$ iff $\mathrm{M}, [a, b]  \models \varphi_1$ and $\mathrm{M}, [a, b]  \models \varphi_2$.
\item $\mathrm{M}, [a, b] \models \langle D\rangle \varphi$ iff there exists an interval $[a', b']$ such that $\mathrm{M}, [a', b'] \models \varphi$, $a \leq a'$, $b' \leq b$, and $[a, b]\neq [a', b']$. In that case we say that $[a, b]$ \emph{ sees} $[a',b']$.
\end{enumerate}

Boolean connectives $\vee, \Rightarrow, \Leftrightarrow$ are introduced in the standard way. 
We abbreviate $\neg \langle D \rangle \neg \varphi$ by $[D]\varphi$ and $\varphi \wedge [D] \varphi$ by $[G] \varphi$. 

A formula $\varphi$ is said to be \emph{satisfiable} in a class of orderings $\cal D$ 
if there exist a structure $\struktura \in \cal D$, a labeling $\gamma$, and an interval $[a, b]$, called {\em the initial interval},
 such that $\langle \odcinki, \gamma\rangle, [a, b] \models \varphi$.
 A formula is satisfiable in a given ordering $\struktura$ if it is satisfiable in $\{\struktura\}$.

\Section{Proof of Theorem \ref{main}}\label{natural}

In Section \ref{natural} only consider finite orderings. 

\noindent{\bf Our representation.} We imagine the Kripke structure of intervals of a finite ordering 
  as a directed acyclic graph, where intervals 
are vertices and each interval $[a, b]$ with the length greater that $0$ has two successors: $[a+1, b]$ and $[a, b-1]$.
 Each level of this representation contains intervals of the same length (see Fig. \ref{exmple}).

\begin{figure}[ht!]
  \centering
    \includegraphics[width=230pt]{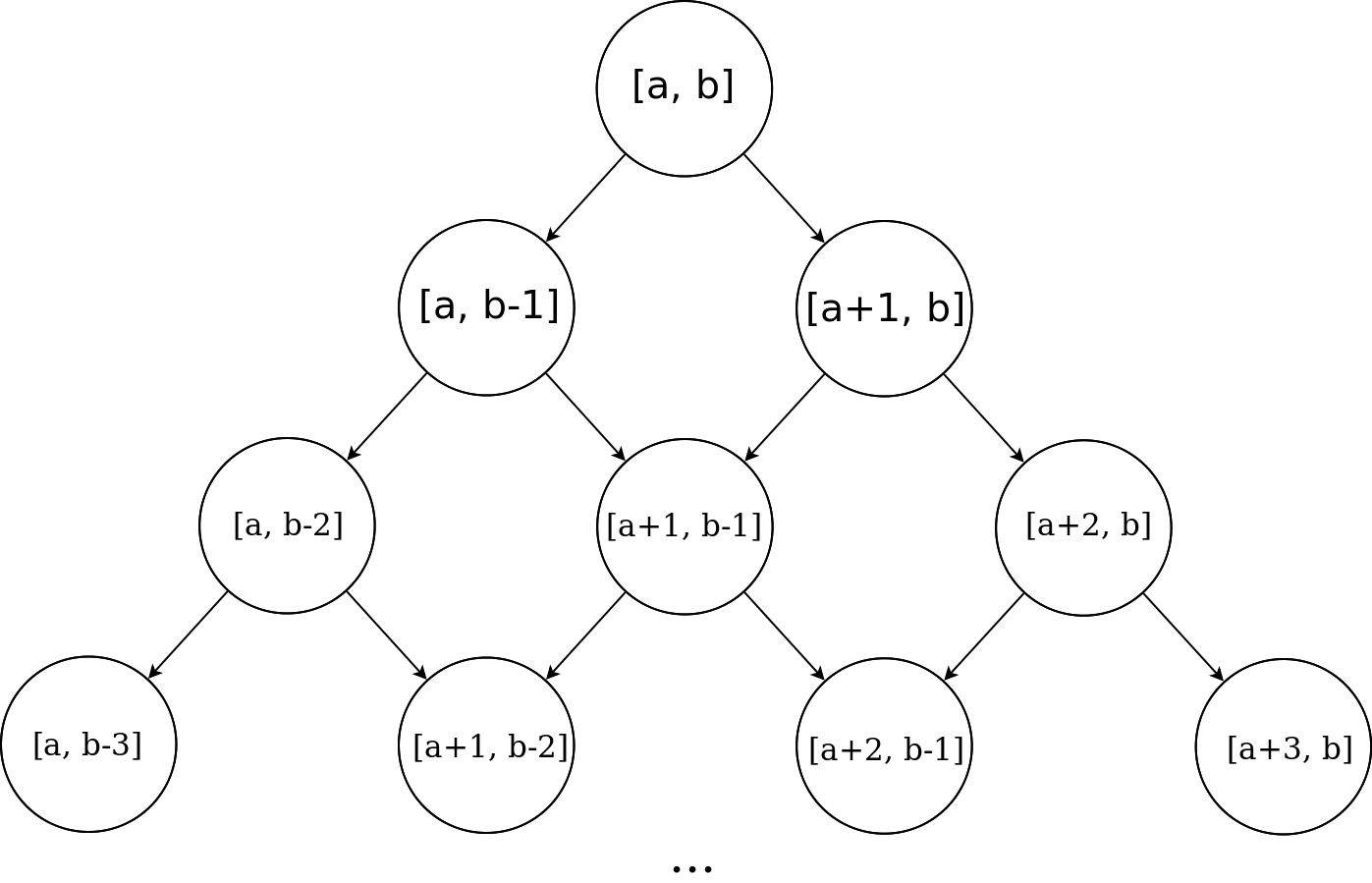}
  \caption{Our representation of order $\langle \{a, a+1, \dots, b\}, \leq \rangle$.}\label{exmple}
\end{figure}

\subsection{The Regular Language $L_{A}$}\label{jezykregularny}

In this section we will, for a given two-counter finite automaton (Minsky machine) $ A$, define a regular language  $L_A$ whose words will  almost\footnote{
See Lemma \ref{prawie} for an explanation what we mean by ''almost''.}
 encode the computation of $ A$ (beginning from the empty counters).

Let $Q$ be the set of states of  $A$, and let $Q'=\{q': q\in Q\} $. 
Define $B=\{f,f_l,f_r,  f',f'_l,f'_r, s,s_l,s_r,s',s'_l,s'_r \} $. 

The alphabet $\Sigma$ of $L_{A}$ will consist of all the elements of $ Q \cup Q'$ (jointly called {\em states}) 
and of all the  subsets (possibly empty) of $B$
which  consist of at most 4 elements: 
at most one of them from  $\{f,f_l,f_r \}$,
at most one from  $\{f',f'_l,f'_r \}$,
at most one from  $\{s,s_l,s_r    \}$,
and at most one from $\{s',s'_l,s'_r \}$.


Symbols of $\Sigma$ containing $f_l$ or $f'_l$ ($s_l$ or $s'_l$)  will be called {\em { f}irst} (resp. {\em { s}econd}) {\em counters}.  
Symbols of $\Sigma$ containing $f_r$ or $f'_r$ ($s_r$ or $s'_r$)  will be called {\em first} (resp. {\em second}) {\em shadows} 
(or {\em shadows of the first/the second counter}).

The language $L_{A}$ will consist of all the words $w$ over $\Sigma$ which satisfy all the following six conditions:

\begin{itemize}
\item The first symbol of $w$ is the beginning state $q_0$ of  $A$ and the last symbol of  
$w$ is either $q$ or $q'$, where $q$ is one of the final states of ${A}$.
\end{itemize}

By a {\em configuration} we will mean a maximal 
sub-word\footnote{By a sub-word we mean a sequence of consecutive elements of a word, an infix.} of $w$, whose first
 element is a state (called {\em the state of this configuration}) and which contains exactly one state (so that $w$ is split into disjoint configurations). 
A configuration will be called {\em even} if its state is from $Q$ and {\em odd} if it is from $Q'$.

\begin{itemize}
\item Odd and even configurations alternate in $w$.

\item Each configuration, except of the last one (which only consists  of the state)
 contains exactly one first counter and exactly one second  counter. If a configuration is even, then its
 first counter contains $f_l$ and its second counter contains $s_l$. If a configurations is odd, then its first counter
contains $f'_l$ and its second counter contains $s'_l$. The first non-state symbol of the first configuration 
is both a first counter and a second counter.

\item The are no shadows in the first and the last configuration. 
Each configuration, except of the first and the last, contains exactly one first shadow and exactly one second shadow.  
If a configuration is even, then its first shadow
contains $f'_r$ and its second shadow contains $s'_r$. If a configurations is odd, then its first shadow
contains $f_r$ and its second shadow contains $s_r$. 
\end{itemize}

It follows, from the conditions above, that if (in a word from the language $L_{A}$) there is a
counter containing $f_l$ ($f'_l, s_l, s'_l$) then there is its shadow $f_r$  (resp. $f'_r, s_r, s'_r$) in the subsequent
configuration. Call a sub-word beginning with first (second) counter and ending with its shadow {\em a first} (resp. {\em second}) {\em shade}. Notice, that the
above conditions imply in particular that each state (except of the first one and last one) is
 in exactly one first shade and  in exactly one second shade.

\begin{itemize}
\item A non-state symbol of $w$ contains $f$ ($f',s, s'$) if and only if it is inside some  shade beginning with $f_l$ (resp. $f'_l,s_l, s'_l$)
\end{itemize}

The last condition defining $L_{A}$  will depend on the  instructions of the automaton $ A$. We say that a configuration has 
{\em first} ({\em second}) {\em counter equal zero} if the first non-state symbol of this configuration contains
$f_l$ or $f'_l$ (resp. $s_l$ or $s'_l$). It is good to think, that the number of symbols before the first/second counter is
the value of this counter in the given configuration. Notice that the first configuration of a 
$w\in L_A$ is indeed the initial configuration of $A$ -- its state is $q_0$ and both its counters equal 0.

 Since the format of an instruction of $A$ is:\\

\noindent
{\tt If in state $q$\\
the first counter\\
equals/does not equal 0 and \\
the second counter\\
equals/does not equal 0\\
then change the state to $q_1$ and\\
decrease/increase/keep unchanged\\
 the first counter and\\ 
decrease/increase/keep unchanged\\
 the second counter.\\} 

\noindent
it is clear what we mean, saying that {\em configuration $C$ matches the assumption of the instruction $I$}.

\begin{itemize}   
\item If $C$ and $C_1$ are subsequent configurations in $w$, and $C$  matches the assumption of an instruction $I$, then:

\begin{itemize}
\item If $I$ changes the state into $q_1$ then the state of $C_1$ is $q_1$.

\item If $I$ orders the first (second) counter to remain unchanged, then the first (resp. second) counter in $C_1$ coincides 
with the first (resp. second) shadow in  $C_1$. 

\item If $I$ orders the first (second) counter to be decreased, then the first (resp. second) counter in $C_1$ is the direct predecessor of 
 the first (resp. second) shadow in  $C_1$. 

\item If $I$ orders the first (second) counter to be increased, then the first (resp. second) counter in $C_1$ is the direct successor of 
 the first (resp. second) shadow in  $C_1$. 
\end{itemize}
\end{itemize}

This completes the definition of the language $L_{A}$. It is clear, that it is regular. Our main tool will be the following:

\begin{lemma}\label{prawie}
The following two conditions are equivalent:

\begin{enumerate}[(i)]
\item
Automaton $A$, started from the initial state $q_0$ and empty counters, accepts.

\item  There exists a 
word $w\in L_{A}$ and a natural number $n$ such that:

\begin{itemize}
\item each configuration in $w$ (except of the last one, consisting of a single symbol)  has length $n-1$

\item each shade in $w$ has length $n$ (this includes the two symbols in the two ends of a shade).
\end{itemize}
\end{enumerate}
\end{lemma} 

\begin{proof} For the $\Rightarrow$ direction consider an accepting computation of $A$ and take $n$ as any number greater than all the numbers that
appear on the two counters of $A$ during this computation. For the  $\Leftarrow$ direction notice that the distance constraint from (ii) imply,
that the distance between a state and the subsequent first (second) shadow equals the value of the first (resp. second) counter in the previous
configuration. Together with the last of the six conditions defining $L_A$ this implies that the subsequent configurations in $w\in L_A$
can indeed be seen as subsequent configurations in the valid computation of $A$.

\end{proof}

Since the halting problem for two-counter automata is undecidable, the proof of Theorem \ref{main} will be completed when we write, 
for a given automaton $A$, a formula $\Psi $ of the language of the  logic of sub-intervals
which is satisfiable (in a finite model) 
if and only if condition (ii) from Lemma \ref{prawie} holds. Actually, what the formula $\Psi $ is
going to say is, more or less, that the word written (with the use of the labeling function $\gamma$) in the leaves of the 
model is a word $w$ as described in Lemma \ref{prawie} (ii).

In the following subsections we are going to write formulae $\Phi_{\mathrm{orient}}$, $\Phi_{L_A}$, $\Phi_{\mathrm{cloud}}$ and $\Phi_{\mathrm{length}}$, such that
$\Phi_{\mathrm{orient}} \wedge \Phi_{L_A} \wedge \Phi_{\mathrm{cloud}} \wedge \Phi_{\mathrm{length}}$ will be the formula $\Psi $ we want.

\subsection{Orientation}\label{orientation}

As we said, we want to write a formula saying that the word written in the leaves of the 
model is the $w$ described in Lemma \ref{prawie} (ii).

The first problem we need to overcome  is the  symmetry of  $D$ -- the operator does not 
see a difference between past and future, or between left and right, so how can we distinguish between the beginning of $w$ and its end?
  We deal with this problem by  introducing five variables $L, R, s_0, s_1, s_2 $ and writing a formula  $\Phi_{\mathrm{orient}}$
which will be satisfied  by an interval $[a, b]$ if 
$[a, a]$ is the only interval that satisfies $L$
and  $[b, b]$  is the only interval that satisfies $R$, or $[b, b]$ is the only interval that satisfies $L$
and  $[a, a]$  is the only interval that satisfies $R$, and  if  all the following conditions hold:

\begin{itemize}
\item any interval that satisfies $L$ satisfies also $s_0$;
\item each leaf is labeled either with $ s_0 $ or with $ s_1$ or with $ s_2 $;
\item each interval labeled with $ s_0 $ or with $ s_1$ or with $ s_2 $  is a leaf;
\item if $c,d,e$ are three consecutive leaves of $[a, b]$ and if $s_i$ holds in $c$, $s_j$ holds in $d$ and $s_k$ holds in $e$ then
$\{i,j,k\}= \{0,1,2\}$.
\end{itemize}

If $[a,b]\models \Phi_{\mathrm{orient}}$ then the leaf of $[a,b]$ where $L$ holds (resp. where $R$ holds) will be called 
the left (resp. the right) end of $[a,b]$. 
%
%

 Let $exactly\_one\_of(X) = \bigvee_{x \in X} (x \wedge \bigwedge_{x' \in {X}\setminus \{x\}} \neg x')$ be a 
formula saying (which is not hard to guess) that exactly
 one variable  from the set  $X$ is true in the current interval. $\Phi_{\mathrm{orient}}$ is a conjunction of the following formulae.

\begin{enumerate}[(i)] 
\item\label{or1} $[D] (([D] \bot \Rightarrow exactly\_one\_of(\{s_0, s_1, s_2\}) \wedge (s_0 \vee s_1 \vee s_2 \Rightarrow [D] \bot) ))$
\item\label{or2} $[D] ( \langle D \rangle \langle D \rangle \top \Rightarrow \langle D \rangle s_0 \wedge \langle D \rangle s_1 \wedge \langle D \rangle s_2)$
\item\label{or4p} $[D] (L \Rightarrow s_0)$
\item\label{or3} $\langle D \rangle R \wedge \langle D \rangle L$
\item\label{or4} $[D] (L \Rightarrow \neg R)$
\item\label{or5} $[D] ([D][D]\bot \wedge \langle D \rangle L \Rightarrow \neg \langle D \rangle s_2)$
\item\label{or6} $\bigvee_{i \in \{0, 1, 2\}} [D] ([D][D]\bot \wedge \langle D \rangle R \Rightarrow \neg \langle D \rangle s_i)$ 
\end{enumerate}

Formulae (\ref{or1}), (\ref{or2}), and (\ref{or4p}) express the property defined by the conjunction of the four items above (notice, that $[D] \bot$ means
that the current interval is a leaf).

Formula (\ref{or3}) says that there exists an interval labeled with $R$ and an interval labeled with $L$.

 Formula (\ref{or4}) states that intervals labeled with $L$  are also labeled with $s_0$, 
and intervals labeled with $R$ are labeled with $s_2$, so they are leaves.

 Formula (\ref{or5}) guarantees that no interval containing exactly 2 leaves,  which is a super-interval of an interval labeled with $L$, 
 can contain a sub-interval labeled with $s_2$. It implies that an interval labeled with $L$ can 
only have one super-interval containing exactly 2 leaves  --- if there were two, then their common super-interval 
containing 3 leaves  would not have a sub-interval labeled with $s_2$, what would contradict  (\ref{or2}).

Finally, formula (\ref{or6}), finally, works like (\ref{or5}) but for $R$. We have to use disjunction
 in this case since we do not know which $s_i$ is satisfied in the interval labeled with $R$.


In the rest of paper we restrict our attention to models satisfying formula $\Phi_{\mathrm{orient}}$, and treat the  leaf labeled with $L$ as the leftmost element
of the model.

\subsection{Encoding a  Finite Automaton}
In this section we show how to make sure that consecutive leaves of the model, read from $L$ to $R$, are labeled with variables that 
represent a word of a given regular language. 

\begin{lemma}
Let ${\cal A} = \langle \Sigma, {\cal Q},  q^0, {\cal F}, \delta\rangle$, 
where $q^0\in {\cal Q}$, ${\cal F} \subseteq {\cal Q}$, $\delta \subseteq {\cal Q} \times \Sigma \times {\cal Q}$ be a finite--state automaton (deterministic or not, it does not matter).
 
There exists  a formula $\psi_{\cal A}$ of the $D$ fragment of Halpern--Shoham logic 
over alphabet ${\cal Q} \cup \Sigma$ that is satisfiable
(with respect to the valuation of the variables from $\cal Q$) if and only if the word,
over the alphabet $\Sigma$
written in the leaves of the model, read from $L$ to $R$, belongs to the language accepted by $\cal A$. 
\end{lemma}

\proof It is enough to write a conjunction of the following properties.
\begin{enumerate}
\item\label{eoaf} In every leaf, exactly one letter from $\Sigma$ is satisfied (so there is indeed a word written in the leaves).
\item Each leaf is labeled with exactly one variable from $\cal Q$.
\item For each interval with the length $1$, if this interval contains an interval labeled with $s_i$,  with $a \in \Sigma$
 and with $q\in {\cal Q}$ and another interval labeled with $s_{(i+1) \texttt{ mod } 3}$, and with  $q' \in {\cal Q}$, 
 then $\langle q, a, q' \rangle \in \delta$.
\item Interval labeled with $R$ is labeled with such  $q\in \cal Q$  and $a \in \Sigma$ that  $\langle q, a, q' \rangle \in \delta$ for some $q'\in \cal F$.
\item\label{eoal} Interval labeled with $L$  is labeled with $q^0$.
\end{enumerate}

Clearly, a model satisfies properties \ref{eoaf}-\ref{eoal} 
if and only if its leaves are labeled  with an accepting run of ${\cal A}$ on the word over $\Sigma $ written in its leaves.
The formulae of the $D$ fragment of Halpern--Shoham logic expressing properties  \ref{eoaf}-\ref{eoal}  are not hard to write:

\begin{enumerate}
\item $[G] (([D] \bot \Rightarrow exactly\_one\_of ({ \Sigma}) ) \wedge (\bigvee \Sigma \Rightarrow [D] \bot) )$
\item $[G] (([D] \bot \Rightarrow exactly\_one\_of ({\cal Q})) \wedge (\bigvee {\cal Q} \Rightarrow [D] \bot))$
\item $[G] ([D][D] \bot \wedge \langle D \rangle s_i \wedge \langle D \rangle s_{i+1 \texttt{ mod } 3} \Rightarrow \bigvee_{\langle q, a, q' \rangle \in \delta} \langle D \rangle (s_i \wedge q \wedge a) \wedge \langle D \rangle (s_{i+1 \texttt{ mod } 3} \wedge q'))$, for each $i \in \{0, 1, 2\}$
\item $[G] (R \Rightarrow \bigvee_{\langle q, a, q' \rangle \in \delta, q' \in \cal F} (q \wedge a))$
\item $[G] (L \Rightarrow q^0)$
\end{enumerate}

Now, let $\cal A$ be a finite automaton recognizing  language $L_A$ from Section \ref{jezykregularny} and put $\Phi_{L_A}=\psi_{\cal A}$.

\subsection{A Cloud -- how to build it}\label{cloud}\label{measuring}


We still need to make sure,  that  there exists $n$ such that each configuration (except of the last one) 
has  length $n-1$  and that each shade has the length exactly $n$. Let us start with:

\begin{definition}
Let $\mathrm{M}=\langle \odcinki, \gamma \rangle$ be a model and $p$ a variable.
We call $p$ {\em a cloud} if there exists $k\in \mathbf{N}$ such that $p\in \gamma([a,b])$ if and only if 
the length of $[a,b]$ is exactly $k$.
\end{definition}

So one can view a cloud as a set of all intervals of some fixed length.
 Notice, that if the current interval has length $k$ then   exactly
$k+1$ leaves are reachable from this segment with the operator $D$.

We want to write a formula of the language $D$ fragment of Halpern-Shoham logic saying that $p$ is a cloud.
In order to do that we use an additional variable $e$. 
The idea is that an interval $[a, a+n]$ satisfies $e$ iff $[a+1, a+n+1]$ does not.

 \begin{figure}[ht!]
  \centering
\includegraphics[width=240pt]{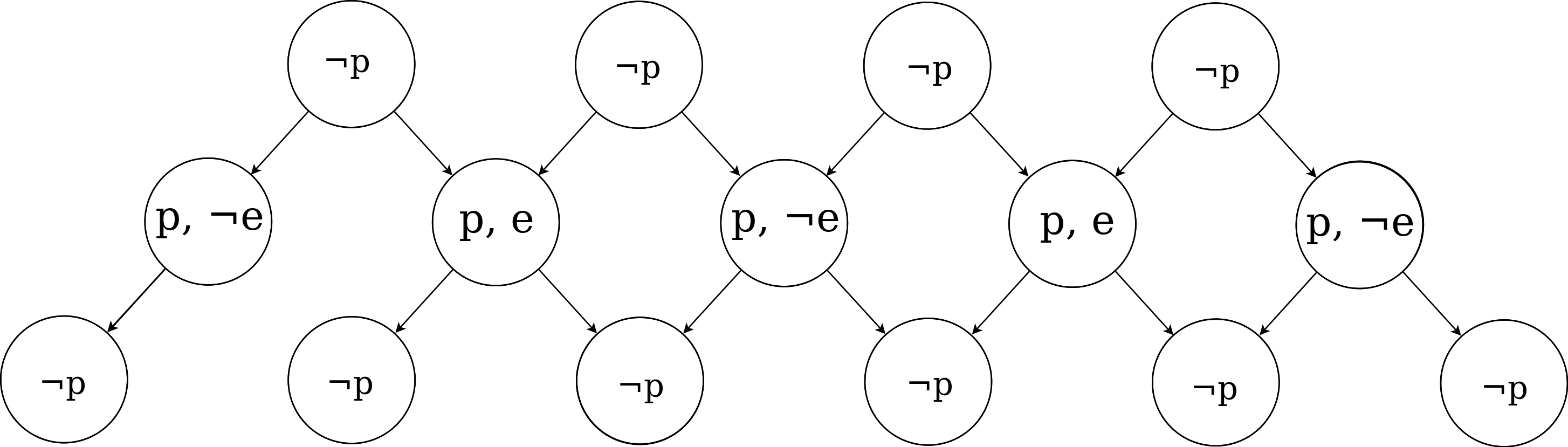}
  \caption{An example of a cloud.}\label{pcloud}
\end{figure}

Let $\Phi_{\mathrm{cloud}}$ be a conjunction of the following formulae.
\begin{enumerate}
\item $\langle D \rangle p$ --- there exists at least one point that satisfies $p$.
\item $[D] (p \Rightarrow [D] \neg p)$ --- intervals labeled with $p$ cannot contain intervals labeled with $p$.
\item $[G] ((\langle D \rangle p) \Rightarrow (\langle D \rangle (p \wedge e)) \wedge (\langle D \rangle (p \wedge \neg e)))$ --- each 
interval that contains an interval labeled with $p$ actually contains at least two such intervals --- one labeled with $e$ and one with $\neg e$. 
\end{enumerate}

\begin{lemma}
If $\mathrm{M}, [a_\mathrm{M}, b_\mathrm{M}] \models \Phi_{\mathrm{cloud}}$, where $a_\mathrm{M}$ and $b_\mathrm{M}$ are endpoints of $\mathrm{M}$, then $p$ is a cloud.
\end{lemma}

\begin{proof}
We will prove that if an interval $[x, y]$ is labeled with $p$, then also $[x+1, y+1]$ is labeled with $p$. 
A symmetric proof shows that the same holds for $[x-1, y-1]$, so 
all the intervals of length equal to $m$, where $m$ is the length of $[x,y]$, are labeled with $p$.
 
This will imply  that no other intervals can be labeled with $p$ and $p$ is indeed a cloud. This is because each such interval either
 has a length greater than $m$, and thus contains an interval of length $m$, and as such labeled with $p$, or
has a length smaller than $m$, and is contained in an interval labeled by $p$, in both cases contradicting (ii).

Consider an interval $[x, y]$ labeled with $p$. Interval  $[x, y+1]$ contains an interval labeled with $p$, 
so it has to contain two different intervals labeled with $p$ -- one labeled with $e$ and the other one with $\neg e$.
 Suppose without loss of generality that $[x, y]$  is the one labeled with $e$, and let us call the second one $[u, t]$. If $t<y+1$, then $[u,t]$ 
 is a sub-interval of $[x, y]$ and is labeled with $p$, a contradiction. So $t=y+1$.
 
 Let us assume that $u>x+1$. The interval $[u-1, y+1]$ must contain two different intervals 
labeled with $p$. One of them is $[x, y+1]$, and it cannot contain another interval labeled with $p$, 
so the other one must be a sub-interval of $[u-1, y]$. But then it is a sub-interval of $[x, y]$ (because $u-1 >x+1-1 = x$) 
which also is labeled with $p$ --- a contradiction. So $u=x+1$.
\end{proof}

\subsection{A Cloud -- how to use it.}

Let us now concentrate on models which satisfy $\Phi_{\mathrm{orient}} \wedge \Phi_{L_A} \wedge \Phi_{\mathrm{cloud}}$. Since $\Phi_{\mathrm{cloud}}$
is satisfied then $p$ is a cloud. Let $n-1$ denote number of leaves contained in the intervals that form the cloud.
Our goal is to write a formula  $\Phi_{\mathrm{length}}$ that would guarantee the following properties:

\begin{enumerate}
\item Configurations and shades are not too short. If you see two states (i.e. more than an entire configuration) or an entire
shade, then you must see a lot, at least $n$ leaves. So you must be high enough. Higher than the cloud. 

\item Configurations and shades are not too long. If you only see an interior of a configuration (i.e. you do not see states)
or an interior of some shade, then you do not see much, at most $n-2$ leaves. So you must be under the cloud. 
\end{enumerate}

Once we do that, the formula  $\Psi =\Phi_{\mathrm{orient}} \wedge \Phi_{L_A} \wedge \Phi_{\mathrm{cloud}}\wedge \Phi_{\mathrm{length}}$ will be 
satisfiable if and only if there exists a word satisfying the  conditions from Lemma \ref{prawie} (ii) -- it is straightforward how
to translate such a word into a model of $\Psi$  and vice versa. 

So put $\Phi_{\mathrm{length}} = \Phi_{\mathrm{length}}^{1, c} \wedge \Phi_{\mathrm{length}}^{1, s} \wedge \Phi_{\mathrm{length}}^{2, c} \wedge \Phi_{\mathrm{length}}^{2, s}$ where:

$\Phi_{\mathrm{length}}^{1, c} = [G] (\bigwedge_{q\in Q, q' \in Q'}  (\langle D \rangle q \wedge \langle D \rangle q') \Rightarrow \langle D \rangle p )$

$\Phi_{\mathrm{length}}^{2, c} = [G] (\bigwedge_{q\in Q}  [D] \neg q \Rightarrow \neg p \wedge [D] \neg p)$

Formulae for shades are a little bit more complex. Let $F_l$ ($F'_l, S_l, S'_l, F, F', S, S', F_r, F'_r, S_r, S'_r$ resp.) be a set of symbols that contain $f_l$ ($f'_l, s_l, s'_l, f, f', s, s', f_r, f'_r, s_r, s'_r$ resp.), and ${\cal T} = \{\langle F_l, F, F_r\rangle, \langle F'_l, F', F'_r\rangle, \langle S_l, S, S_r\rangle , \langle S'_l, S', S'_r\rangle \}$. 

\[
\begin{array}{@{}r@{}c@{}l@{}}
\Phi_{\mathrm{length}}^{1, s} &=&  [G](\bigwedge_{\langle T_l, T, T_r \rangle\in {\cal T}} (\langle D \rangle \bigvee T_l \wedge \langle D \rangle \bigvee T_r) \Rightarrow \langle D \rangle  p)\\

\Phi_{\mathrm{length}}^{2, s} &=&\displaystyle [G](\bigwedge_{\langle T_l, T, T_r\rangle \in {\cal T}} (\langle D \rangle \bigvee T \wedge \neg \langle D \rangle \bigvee (T_l \cup T_r))\\

&& \Rightarrow \neg p \wedge [D] \neg  p)
\end{array}
\]

\Section{Proof of Theorem \ref{main-discrete}}\label{th2}

\todo{We can lift a cloud!}

\clearpage	
\newpage	
\end{document}